\documentclass[11pt]{scrartcl}
\usepackage[latin1]{inputenc}
\usepackage[english]{babel}
\usepackage{amsmath,amsthm,amssymb,bbm,times,color}

\theoremstyle{plain}
\newtheorem{theorem}{Theorem}[section]
\newtheorem{lemma}[theorem]{Lemma}
\newtheorem{corollary}[theorem]{Corollary}
\newtheorem{proposition}[theorem]{Proposition}

\theoremstyle{definition}

\theoremstyle{remark}

\numberwithin{equation}{section}

\newcommand{\N}{\mathbb{N}}
\newcommand{\R}{\mathbb{R}}
\newcommand{\C}{\mathbb{C}}

\DeclareMathOperator{\tr}{Tr}
\DeclareMathOperator{\sign}{sign}
\DeclareMathOperator{\adj}{adj}
\DeclareMathOperator{\dist}{dist}
\DeclareMathOperator{\infspec}{inf\, spec}

\begin{document}

\title{Decay of correlations and absence of superfluidity in the disordered Tonks-Girardeau gas}
\author{Robert Seiringer$^1$ and Simone Warzel$^2$  \vspace{.3cm} \\ \vspace{-.15cm}
{\small $^1$IST Austria, 
Am Campus 1, 3400 Klosterneuburg, Austria} \\ 
{\small $^2$Zentrum Mathematik, TU M\"unchen, 
Boltzmannstr. 3, 85747 Garching, Germany}}
\date{January 26, 2015}							

\maketitle

\begin{abstract}
{\bf Abstract.} We consider the Tonks-Girardeau gas subject to a random external potential. If the disorder is such  that the underlying one-particle Hamiltonian displays  localization (which is known to be generically the case), we show that there is exponential decay of correlations in the many-body eigenstates. Moreover, there is no Bose-Einstein condensation and no superfluidity, even at zero temperature. 

\end{abstract}

\section{Introduction}
Understanding the various aspects and even the qualitative structure of phase diagrams of interacting many-body systems in the  presence of static disorder still poses a big challenge. 
Basic questions, such as the existence and characterizations of a phase of many-body localized states, remain under debate --- even for one-dimensional systems (cf.~\cite{AV,NH} and references therein). 
For bosons, one manifestation of localization is the existence of a glass phase in which the  static correlations decay and superfluidity is absent  \cite{FWGF89}. While such a phase is predicted to exist for strong interactions or strong disorder, for intermediate interaction strength superfluidity is expected to persist at small values of the disorder even in one dimension \cite{GSch,RappSch}. The interest in these questions was renewed due to experimental accessibility of such systems \cite{LAexp} (see also \cite{SPReview,CCG+11} and references therein). 

In this context, and in view of the woefully short list of rigorous results on disordered systems with interaction~\cite{Imbrie,Mastr}, limiting or integrable model systems   present a testing ground for numerical works, conjectures and ideas (cf.~\cite{SYZ,BW,KMSY,KV14,SS}). 
In the bosonic case, the limiting case of hard-core repulsive interaction is such an example: in the lattice set-up this amounts to studying the $XY$-spin Hamiltonian with a random magnetic field, and in the continuum this is the Tonks-Girardeau model with a random potential, which is the topic of the present paper.  Such hard-core interactions may actually be realized experimentally~\cite{PWexp,KWW} --- albeit without disorder. 
Both models can be related to non-interacting fermions in an external random potential. They are not exactly solvable, but nevertheless amenable to rigorous analysis; the difficulty in both cases lies in the non-local dependence of the physical (bosonic) correlation functions on the underlying fermionic correlation functions. 
In \cite{Eggeretal} this link was exploited numerically to show that the disorder destroys bosonic quasi-long-range order. For the $XY$-model such a result can be confirmed by rigorous bounds on the correlations of any eigenstate \cite{SW} (see also \cite{KP,HSS12,AS,ASSN} for related and earlier results in this context). 
The purpose of this paper is to show that such results also apply to the corresponding continuum model. In addition to a proof of the exponential decay of correlations for all eigenstates, we show that the superfluid density (or stiffness)  vanishes (exponentially) at zero temperature. Our basic assumption in all these results is the (exponential) localization of the underlying one-particle operator --- a property which generically holds true up to arbitrarily large energies in one-dimensional disordered systems \cite{LGP,PF}. \\

We consider a system of $ N $ bosons with point interactions on a ring with length $ L $, which we take to be an integer for simplicity. It is described by a many-particle Hamiltonian of the form
\begin{equation}
\mathcal{H}_{L,\omega} 
= \sum_{j=1}^N \big(H_{L,\omega}^+ \big)_j 
+ g \!\sum_{1\leq j<k\leq N} \delta(x_j-x_k)
\, . 
\end{equation}
We will be interested in the cases when the  
one-particle Hamiltonian is given by
\begin{equation}\label{def:Ham}
H^+_{L,\omega}
:= - \frac{d^2}{dx^2}
+ V_\omega (x)
\end{equation}
on $ L^2([0,L]) $ with periodic boundary conditions. 
The dependence on $\omega$, which will often be omitted from the notation for convenience, indicates the randomness entering the potential landscape. We will 
assume throughout that the following probabilistic average is finite,
\begin{equation}\label{eq:assv}
\sup_n\,  \mathbb{E}\left[ \int_{I_n} |V(x)| dx \right] < \infty \, , \qquad I_n :=(n-1,n] \, ,
\end{equation}
where $\mathbb{E}$ stand for the expectation with respect to $\omega$. 
This ensures, in particular, that $ V_\omega  \in L^1([0,L]) $ and that the one-particle Hamiltonian \eqref{def:Ham} (defined via its quadratic form on the Sobolev space $ H^1[0,L]$) is self-adjoint in $ L^2([0,L]) $ (with any self-adjoint boundary conditions and, in particular, with periodic ones) for almost all $ \omega $. 

In the Tonks-Girardeau limit $ g \to \infty $, the bosonic wavefunctions are required to vanish upon particle contact,  i.e., $ \Psi(x_1,\dots, x_N) = 0 $ in case $ x_j = x_k $ for some $ j \neq k $.
Any eigenfunction of $ \mathcal{H}_{L,\omega} $ hence takes the form of an eigenfunction of a system of $ N $ non-interacting fermions multiplied by a suitable phase to render it symmetric upon 
particle exchange \cite{Gir,GirYuk}. More precisely, let $ H_L^\pm $ stand for \eqref{def:Ham} with periodic ($+$) or anti-periodic ($-$) boundary conditions (b.c.).  If $\{  \varphi_{j,L}^\pm \} $ denotes an eigenbasis of $ H_L^\pm $ and $ \{ j_\alpha \}_{\alpha=1}^N $ indexes a subset of $ N $ orthonormal eigenfunctions, then 
\begin{equation}\label{eq:Bosewf}
\Psi(x_1,\dots, x_N) = \frac{1}{\sqrt{N!}} \det\left( \varphi_{j_\alpha,L}^{\sharp_N}(x_\beta) \right)_{\alpha,\beta=1}^N \,  \prod_{1\leq j < k \leq N} \sign (x_j - x_k)  
\end{equation}
is a normalized eigenfunction of the Tonks-Girardeau Hamiltonian provided we choose $ \sharp_N := (-1)^{N+1} $, i.e., {\em anti-periodic} b.c. in case $ N $ is even and  {\em periodic}  b.c. in case $ N $ is odd (cf. \cite{LieLin}). 
In particular, the bosonic ground state of $ \mathcal{H}_L $
corresponds to choosing $ \{j_\alpha \}_{\alpha=1}^N $   the $ N $ lowest eigenvalues of $ H_L^{ \sharp_N} $, and its ground state energy $E_L(N)$ is simply the sum of the lowest $N$ eigenvalues of $H_L^{ \sharp_N}$. \\

We will mainly investigate two quantities of interest:
\begin{enumerate}
\item
The {\bf reduced one-particle density matrix} $ \gamma_\Psi $ corresponding to any eigenstate $ \Psi $ given by~\eqref{eq:Bosewf}. It is defined through its kernel
\begin{equation}\label{eq:1PD}
\gamma_\Psi(x,y) := N \int  \Psi(x,x_2,\dots,x_N) \overline{\Psi(y,x_2,\dots,x_N)} \, d x_2 \dots d x_N \, , 
\end{equation}
and satisfies $ 0 \leq \gamma_\Psi  \leq N $ and $ \tr \gamma_\Psi  = N $. Bose-Einstein condensation (BEC) refers to a macroscopic value of the largest eigenvalue $\| \gamma_\Psi  \| $ of this operator  in the thermodynamic limit ($N,L \to \infty $ with $ \frac{N}{L} = {\rm const.}$) \cite{PO} (see also~\cite{LSSY}).

It is not hard to show that the reduced one-particle density matrix takes the form of a determinant of an $ (N+1) \times (N+1) $ block matrix
\begin{equation}\label{eq:opddet}
\gamma_\Psi(x,y) = \det \begin{pmatrix} 0 & \varphi_{j_1,L}^{\sharp_N}(x) \cdots \varphi_{j_N,L}^{\sharp_N}(x) \\ \begin{matrix}  \overline{ \varphi_{j_1,L}^{\sharp_N}(y) } \\  \vdots \\  \overline{\varphi_{j_N,L}^{\sharp_N}(y) } \end{matrix} &  K_N(x,y) \end{pmatrix}
\end{equation}
where the $ N \times N $-submatrix $  K_N(x,y) $ is for all $ x \leq y $ given by the entries
\begin{align}\label{def:KN}
[ K_N(x,y)]_{\alpha,\beta} \ & := \delta_{\alpha,\beta} - 2 \int_{[x,y]} \varphi_{j_\alpha,L}^{\sharp_N}(z) \overline{\varphi_{j_\beta,L}^{\sharp_N}(z) } dz \notag \\
& = - \delta_{\alpha,\beta} + 2 \quad  \int_{[x,y]^c}  \varphi_{j_\alpha,L}^{\sharp_N}(z) \overline{\varphi_{j_\beta,L}^{\sharp_N}(z) } dz \, .
\end{align} 
Introducing the projection $ P_N = \sum_{\alpha } |  \varphi_{j_\alpha,L}^{\sharp_N} \rangle \langle  \varphi_{j_\alpha,L}^{\sharp_N} | $ onto the eigenfunctions entering the state $ \Psi $, we may write
$ K_N(x,y) = P_N - 2 P_N 1_{[x,y]} P_N $ 
as an operator identity on $ P_N L^2([0,L]) $. In this manner, one easily sees that $ \gamma_\Psi $ only depends on the projection $ P_N $, as a change of basis corresponds to a unitary transformation of the matrix in \eqref{eq:opddet} which leaves the determinant invariant. 
\item The {\bf  superfluid density} (or:  {\bf stiffness}) measures the extent to which the ground state energy of the Tonks-Girardeau Hamiltonian increases as one twists the boundary conditions \cite{FBJ}, i.e., when the wave-functions are required to pick up a phase $e^{i\theta}$ as one particle moves around the ring,
$ \Psi(x_1,\dots, x_j+L,\dots , x_N) = e^{i\theta}\, \Psi(x_1,\dots, x_j,\dots , x_N) $.
In other words, the superfluid density $\rho_s$ is defined via the ground-state energy shift $E_L(N,\theta) \approx E_L(N,0) + \theta^2 \rho_s/L$ for small $\theta$. To give a precise definition in the thermodynamic limit, we find it more convenient to work in a grand-canonical  picture where the particle number is determined by a fixed chemical potential $\mu$; i.e., $\mu$ is chosen independently of $N$ and $\omega$.

For given $\mu\in \R$, let $N^\pm_\mu := \tr 1_{(-\infty,\mu]}(H_L^\pm)$ denote the number of eigenstates of $H_L^\pm$ below $\mu$ and set $N_\mu := \min\{ N^+_\mu, N^-_\mu\}$. With $E_L(N,\theta)$ the ground state energy of the Tonks-Girardeau Hamiltonian with twisted boundary conditions, the superfluid density is defined as 
\begin{equation}\label{eq:sfd}
 \rho_s := \limsup_{\theta \to 0 } \frac{1}{ \theta^2} \limsup_{L\to \infty} L \left( E_L(N_\mu, \theta) - E_L(N_\mu,0)  \right) \, . 
\end{equation}
We note that it follows from the diamagnetic inequality \cite[Thm.~7.21]{LiebLoss} that $E_L(N, \theta) \geq E_L(N,0)$ for any $\theta$ and $N$, hence $\rho_s\geq 0$. 

With our definition of $N_\mu$, the ground state energy $E_L(N_\mu,0)$ has the following convenient representation in terms of $H^\pm_L$. 
  With $\sharp_\mu :=  \sharp_{N_\mu } = (-1)^{N_\mu +1}$, we have
\begin{equation}\label{eq:remnu}
E_L(N_\mu,0) = \mu N_\mu - \tr [H_L^{\sharp_\mu} -\mu]_- = \sum_{j: E^{\sharp_\mu}_{j,L} \leq \mu}  E_{j,L}^{\sharp_\mu} \,,
\end{equation}
where  the $\{ E^\pm _{j,L} \} $ denote the eigenvalues of $ H^\pm_L$, and $[\, \cdot \,]_- := - \min\{ 0, \,\cdot\,\}$ denotes the negative part. In other words, for any $\mu\in \R$, $H^{\sharp_\mu}_L$ has exactly $N_\mu$ eigenvalues below $\mu$. This is a consequence of the fact that $N_\mu = N^{\sharp_\mu}_\mu$, which, in turn, follows from the interlacing property $E_{j,L}^+ < E_{j,L}^-$ for $j$ odd, and $E_{j,L}^+ > E_{j,L}^-$ for $j$ even (see  \cite[Thm 2.3.1]{East} or \cite[Thm.~XIII.89]{RS4}).

\end{enumerate}

\section{Results} 
\subsection{Localization hypothesis and first consequences}\label{sec:loc}

We will assume that the one-particle operator $ H^\pm_{L,\omega}$ in \eqref{def:Ham}, for both periodic ($+$) and anti-periodic ($-$) b.c.,   exhibits (sub-)exponential Anderson localization  with some exponential parameter $ \xi \in (0,1] $ and localization length $ \ell < \infty $ in the energy regimes of interest. To be more specific, we consider the eigenfunction correlator
\begin{equation}
Q_L^\pm(n,m;J;\omega) := \sum_{j:\, E^\pm_{j,L} \in J } \Phi_{j,L}^\pm(n;\omega)  \Phi_{j,L}^\pm(m;\omega) 
\end{equation}
corresponding to some  energy regime $ J \subset \mathbb{R} $. 
Here 
\begin{equation}
\Phi_{j,L}^\pm(n;\omega) := \left(\int_{I_n}  \left| \varphi_{j,L}^\pm(x;\omega)\right|^2 dx\right)^{\frac{1}{2}}
\end{equation}
quantifies the probability for the $ j $th eigenfunction to be present on a basic interval of unit length. 
We shall tacitly assume that the complete orthonormal set of eigenfunctions of $ H^\pm_{L,\omega} $ is jointly measurable in $ (x, \omega) $. 
(In case of degeneracy, which generically is believed to be absent with probability one, this in particular requires the choice of a proper labelling of eigenfunctions.)

\begin{description}
\item[Localization hypothesis (ECL) on $ J$:] There exist $ C, \ell \in (0,\infty) $ and $\xi\in(0,1]$ such that for all $ 1\leq n , m \leq L $ and all $L\in \N$
\begin{equation}\label{eq:ECL}
 \mathbb{E}\left[Q_L^\pm(n,m;J)\right] \leq C \, \exp\left(-\frac{\dist(n,m)^\xi}{\ell^\xi} \right)\,,
\end{equation}
where $ \dist(\cdot,\cdot)$ denotes the Euclidean distance on the (one-dimensional) torus.
\end{description}

 In the theory of (one-particle) random operators, the condition (ECL) is both strong and convenient: it ensures localization in both the spectral sense (i.e. only pure-point spectrum in $ J$ with (sub-)exponentially decaying eigenfunctions) as well as in the strong dynamical sense that 
\begin{equation}\label{eq:dynloc}
 \mathbb{E}\left[ \sup_t \| 1_{I_n} e^{-itH_L^\pm} P_J(H_L^\pm)  1_{I_m} \|_1 \right] \leq C  \exp\left(-\frac{\dist(n,m)^\xi}{\ell^\xi} \right) \, .
\end{equation}
Here $ P_J(H_L^\pm) $ denotes the spectral projection of $ H_L^\pm $ onto the energy regime $ J $ and $ \| \cdot \|_1 := \tr | \cdot | $ is the trace norm. 
Eigenfunction correlator localization (ECL) is established for a large class of single-particle random Schr\"odinger operators by means of either the continuum fractional-moment method \cite{amcont} (which is based on \cite{aizmolch}) or the via the bootstrap multi-scale analysis \cite{GK1,GK2} (which is based on \cite{froespe}). In our one-dimensional set-up, localization is expected to hold generically at all energies. In particular, (ECL) will hold for 
$J = (-\infty , \mu]$ with $ \xi=1 $ and some localization length $ \ell = \ell_\mu $ which depends on $ \mu \in \R $ only. This has been established in the following specific
models:
\begin{itemize}
\item early on in the history of localization proofs \cite{GMP} for random potentials of the form $ V_\omega (x) = F(b_x(\omega)) $ with  $ b_x(\omega) $ a Brownian motion on a compact Riemannian manifold $ M $ and $ F:M \to \mathbb{R} $ a smooth Morse function with $ \min_M F = 0 $.
\item for homogeneous {\em Gaussian random potentials}, i.e., $ \mathbb{E}[V(x)] = 0 $  with covariance function $ C(x-y) := \mathbb{E}[V(x) V(y) ]$, which admits the representation $ C(x ) = \int \gamma(x+y) \gamma(y) dy $ in terms of a non-negative, compactly supported function $ \gamma $, which is uniformly H\"older continuous, i.e., there is   $ s \in (0,1] $ and $ a < \infty $ such that $ |  \gamma(x+y) - \gamma(x) | \leq a |y|^s $ for all $ x $ and all $ y > 0 $ sufficiently small (cf. \cite{FLM,Ueki}). 
\item for {\em alloy-type random potentials} $ V_\omega(x)  = W(x) +  \sum_{j \in \mathbb{Z}} \omega_j U(x-j ) $ with independent and identically distributed random variables $ (\omega_j)_{j\in \mathbb{Z}} $ whose distribution is absolutely continuous with a bounded density, i.e. $ \mathbb{P}( \omega_j \in dv) = \varrho(v) dv $ with some $ \varrho \in L^\infty(\mathbb{R}) \cap L^1(\mathbb{R}) $ of compact support. The term $ W $ serves as a non-random, bounded, $ 1$-periodic background potential and the single-site potential $ U $ is assumed to satisfy $ c 1_I(x) \leq U(x) \leq  C 1_{[0,1]}(x) $ for some $ c , C \in (0,\infty ) $ and a non-trivial sub-interval $ I \subset [0,1] $ (cf.~\cite{HSS}). 
\end{itemize}

ECL implies the (sub-)exponential localization of eigenfunctions about some random localization center. More precisely, it implies what is called semi-uniform localization of eigenfunctions (SULE). For the definition of the latter it is convenient to fix a weight function $ g_L : \{ 1,\dots ,L\} \to [1,\infty) $ with the property $ \sum_{\alpha=1}^L g_L(\alpha)^{-1} = 1 $. A specific choice, which we will adopt below,  is $ g_L(\alpha) = L $. 
\begin{description} 
\item[Localization hypothesis (SULE) on $ J $:] There exist  $\ell \in (0,\infty)$, $\xi\in (0,1]$ and, for every $ L \in \mathbb{N} $, an amplitude $ A_{L,\omega} \geq 0$ that is uniformly integrable, i.e., 
\begin{equation}\label{eq:unifA}
\sup_{L\in \mathbb{N} } \ \mathbb{E}\left[A_L \right] < \infty \, , 
\end{equation} 
such that  for every eigenfunction $ \varphi_{j,L}^\pm$ of $ H^\pm_{L,\omega} $ with eigenvalue $E_{j,L}^\pm \in J  $ there is some $ \gamma^\pm_{j,L,\omega} \in \mathbb{N} $ such that for all $ n $:
\begin{equation}\label{eq:locass}
 \Phi_{j,L}^\pm(n;\omega)  \leq A_{L,\omega} \, g_L(\gamma^\pm_{j,L,\omega})^{3/2} \, \exp\left(-\frac{\dist(n,\gamma_{j,L,\omega}^\pm)^\xi}{\ell^\xi} \right) \, . 
\end{equation}
\end{description}
The points $ \gamma_{j,L,\omega}^\pm$ play the role of localization centers. However, they need not coincide with the location of the maxima of $  \Phi_{j,L}^\pm(n;\omega) $. 
The length $ \ell $ is non-random and coincides with the minimum of all localization length at energies in $ J $. 
At first sight, the role of the function $ g_L $ might be puzzling and one may be tempted to drop the factor $ g_L^{3/2} $ on the right side of \eqref{eq:locass}. This, however, is known to be wrong \cite{DelRio95}. 
 If assumption (ECL) holds for an energy regime $ J$, then (SULE) holds with any weight function $ g_L $ for the same energy regime (but possibly with a slightly reduced localization length), cf.~\cite[Ch.7]{AiWa15}.\\

 As  explained in the introduction, every fermionic many-body state $ \phi(x_1,\dots, x_N) $, which is either periodic or antiperiodic  depending on whether $ N $ is odd or even, gives rise to the periodic bosonic many-body state $ \Psi(x_1,\dots,x_N) =  \phi(x_1,\dots, x_N) \,  \prod_{1\leq j < k \leq N} \sign (x_j - x_k)   $. In the Tonks-Girardeau limit, the dynamics of such a state is given in terms of the dynamics of free fermions, i.e.,  $ \phi_t(x_1,\dots, x_N) = \exp\left( -i t \sum_{j=1}^N  (H_L^{ \sharp_N})_j \right) \phi(x_1,\dots, x_N) $.
 While the bosonic one-particle density matrix $ \gamma_{\Psi_t} $ of this state does not coincide with the fermionic  one, given by $ \Gamma_{ \phi_t} $ (which is defined as in~\eqref{eq:1PD} with $ \Psi $ replaced by $ \phi_t $)
 ,  their diagonals  agree, i.e., the 
 bosonic and fermionic densities are equal:
\begin{equation}
\varrho_t(x) :=  \gamma_{\Psi_t}(x,x) = \Gamma_{ \phi_t}(x,x) = \left( e^{-it H_L^{ \sharp_N}}  \Gamma_{ \phi} e^{it H_L^{ \sharp_N}} \right)(x,x) \, .
 \end{equation}
Dynamical localization for free fermions, in the form \eqref{eq:dynloc},  then immediately entails the following result for any many-body eigenstate. The bound~\eqref{eq:O1change} is a manifestation of many-body localization for the model of interacting bosons considered here.
  \begin{proposition}\label{prop:dyn}
 If the range of $  \Gamma_{ \phi}  $ at $ t = 0 $ falls within a regime of  dynamical localization, i.e.  $ \Gamma_{ \phi} = P_J( H_L^{ \sharp_N}) \Gamma_{ \phi} = \Gamma_{ \phi} P_J( H_L^{ \sharp_N})  $ for some $ J \subset \R $ for which~\eqref{eq:dynloc} holds, then there exists an  $ A \in (0,\infty) $ which is independent of $ L $ and $N$ such that:
 \begin{enumerate}
 \item the total number of particles on any subset $ I \subset [0,L] $ changes on average by order one only:
 \begin{equation}\label{eq:O1change}
\mathbb{E}\left[\sup_{t\in \R} \left| \int_I \varrho_t(x) dx -  \int_I \varrho_0(x) dx \right| \right] \leq  A \quad  \text {for all $I\subset [0,L]$}\,,
\end{equation}
 \item for any pair of subsets $ I \subset  K \subset [0,L] $:
 \begin{equation}\label{eq:nostatesrem}
 \mathbb{E}\left[ \sup_{t\in \mathbb{R}} \int_I \varrho_t(x) dx \right] \leq  \mathbb{E}\left[ \int_K \varrho_0(x) dx \right] + A \, \exp\left( - \frac{\dist(I, K^c)^\xi}{\ell^\xi} \right) \, . 
 \end{equation}
 \end{enumerate}
  \end{proposition}
The (simple) proof of this proposition will be given in Appendix~\ref{app}. 
 Both statements are expressions of non-ergodic behavior of the localized system: One may prepare the system initially in a state which exhibits a step-like profile in its density, i.e., some positive averaged density in one half ($I$) and another one in the other half ($I^c$). In such a situation, \eqref{eq:O1change} states that the step-like profile remains for arbitrarily long times with only a finite number of particles crossing on average.
 
 The second bound 
is relevant for experiments in which the bosons are initially trapped around some location (such that $ \int_K \varrho_0(x) dx \approx 0 $) and then released from the trap at $ t= 0 $. The localization bound~\eqref{eq:nostatesrem} then guarantees that the total number of particles will remain small on average away from the initial location, uniformly in time (confirming numerical simulations in \cite{RJSB}.) A related bound for the XY-model can be found in~\cite[Thm.~1.1]{ASSN}.

\subsection{Decay of correlations and absence of BEC}
Our first non-trivial consequence of one-particle localization concerns a strong version of absence of off-diagonal long-range order (ODLRO).

\begin{theorem}\label{thm:1pd}
Let $ \Psi $ be any  many-particle eigenstate of the form~\eqref{eq:Bosewf} which is composed of a selection of one-particle states $ \{\varphi_{j_\alpha,L}^{\sharp_N}\}_{\alpha = 1}^N $ corresponding to an energy regime $ J $. 
If condition (ECL) holds for $J $, then there exist  $A \in (0,\infty) $  independent of $L$ and $N$ such that
\begin{equation}\label{eq:1pd}
 \mathbb{E}\left[ \| 1_n \gamma_\Psi 1_m \|_2^{\sigma}\right] \leq A \,   \exp\left(-\frac 23(1-\sigma) \frac{\dist(n,m)^\xi}{(2\ell)^\xi} \right)
\end{equation}
for all $1\leq n , m\leq L $ and all $2/5\leq \sigma < 1 $. Here $ \| \cdot \|_2 $ denotes the Hilbert-Schmidt-norm on $ L^2([0,L]) $.
\end{theorem}

A proof of this theorem, as well as of the subsequent corollary, will be given in Section~\ref{sec:DecCor}. The proof shows that the result can easily be extended in various directions, and is not restricted to eigenstates of the many-particle Hamiltonian. It applies, e.g., to general states of the form  \eqref{eq:Bosewf} as long as the  one-particle functions $\varphi_{j}$ are suitably localized, and is thus also relevant in time-dependent situations as in \cite{RM04}.  \\

Absence of BEC is not immediately implied by the absence of ODLRO, since our assumptions on the system allow for unbounded fluctuations of the density. We therefore need a mild additional assumption on these fluctuations in order to reach such a conclusion.

\begin{corollary}\label{cor:nobec}
Given the assumptions of Theorem~\ref{thm:1pd}, assume additionally that for $ p > 2 $
\begin{equation}\label{eq:asscor}
\sup_{n,L}\, \mathbb{E}\left[ \left( \tr 1_{I_n} P_J(H_L^\pm) \right)^{p} \right] < \infty \, . 
\end{equation}
Then 
for  any sequence $ \Psi $ of eigenstates composed of one-particle states $ (\varphi_{j_\alpha,L}^{\sharp_N})_{\alpha = 1, \dots, N} $ 
whose energies fall into a regime $ J $ of (ECL), 
the almost-sure convergence 
\begin{equation}\label{eq:nobec}
\lim_{L \to \infty} \frac{\|  \gamma_{\Psi} \|}{L^r} = 0 
\end{equation}
holds for any $ \frac{2}{p} < r \leq 1 $. 
\end{corollary}

Note that the convergence \eqref{eq:nobec} is independent of the choice of $N$, which is allowed to depend on $L$ and $\omega$. Typically one is interested in the case that $N \approx {\rm const.}\, L$ as $L\to\infty$. 
The subsequent proof (specifically Eq.~\eqref{eq:boundp}) also shows that in case~\eqref{eq:asscor} holds with $ p > 1 $ the (averaged) momentum distribution associated with the state $ \Psi $,  
\begin{equation}  
n(k) := \frac{1}{L} \iint e^{ik(x-y)} \mathbb{E} \left[ \gamma_\Psi(x,y) \right]  dx dy \, , 
\end{equation}
remains uniformly bounded, since $ | n(k) | \leq \sup_L \frac{1}{L} \sum_{n,m } \mathbb{E}\left[ \| 1_n \gamma_{\Psi} 1_{I_m} \|_2 \right] < \infty $.  (In particular, $ n(0) < \infty $. At large values $ |k|\to \infty $, one expects an algebraic fall-off $ n(k) \sim k^{-4} $ due to the hard-core repulsion \cite{BZ11}). This is consistent with numerical predictions in~\cite{Eggeretal}.\\

A simple sufficient condition for  \eqref{eq:asscor} to hold for any value $ p > 2 $ (and hence for the validity of \eqref{eq:nobec} for any $ r > 0 $) and  $ \mathbb{Z} $-homogeneous random potentials in case $ J \subset (-\infty, \mu] $  is the existence of an exponential moment,
\begin{equation}
\int_{I_1} \mathbb{E}\left[e^{-t V(x )} \right]  dx < \infty 
\end{equation}
for some $ t > 0 $. (A derivation of this statement starts from the estimate $ \tr 1_{I_n} P_{ (-\infty, \mu] }(H_L^\pm)  \leq e^{t\mu} \tr 1_{I_n} e^{-tH_L^\pm} $ and proceeds through standard semigroup bounds, cf.~\cite[Ch.II 5]{PF}.) This is clearly satisfied for the models listed above. (Alternatively, one may proceed though resolvent techniques as in \cite[Ch.II 5]{PF} to show that the finiteness  of $\sup_n \int_{I_n} \mathbb{E}[ |V(x)|^{2p} ] \,  dx  $ is sufficient for \eqref{eq:asscor} to hold.) \\

The absence of ODLRO and BEC  is not a consequence of the disorder alone: In case $ V = 0 $, Lenard~\cite{Len} showed that the reduced density matrix of the ground state wave function $ \Psi $ behaves as $ \gamma_\Psi(x,y)\sim  |x-y|^{-1/2}$  for large $ |x-y | $. This slow fall off is sometimes referred as quasi-long range order and causes of the order of $ \sqrt{N} $ particles to quasi-condense into the zero mode. 

The above corollary  (with $ r < 1/2 $) shows that localization decreases the rate of quasi-condensed particles in comparison to the free case ($V=0$). This should not be taken for granted as a comparison with the non-interacting case shows. 
For {\em non-interacting} bosons in a non-negative Poisson random potential, Luttinger together with Kac  \cite{KaLu73,KaLu74} and Sy \cite{LuSy73b} noted that 
that critical dimension $ d $ for the occurrence of BEC is lowered to $ d = 1 $. (A rigorous version of their analysis is contained in~\cite{LePaZa04,LZ07} and the basic mechanism also applies to alloy-type random potentials.) The occurrence of BEC in an ideal Bose gas even at positive temperature is due to the behavior of the density of states  near the bottom of the one-particle energy spectrum. 
The latter is severely suppressed due to the occurrence of Lifshitz tails, which causes a macroscopic fraction of the particles to condense into modes whose energy vanishes in the thermodynamic limit. 
Since Anderson localization is known for such models,  our results imply that the interactions destroy BEC, and the corresponding 
Tonks-Girardeau model shows no BEC even at zero temperature.

\subsection{Absence of superfluidity}

In the absence of an external random potential ($V=0$), it is well-known that the superfluid density~\eqref{eq:sfd} coincides with the total density at chemical potential $ \mu $:
\begin{equation}
\rho_s = \lim_{L\to \infty} \frac{1}{L} N_\mu = \frac{\sqrt{[\mu]_+}}{2\pi^2} \,.
\end{equation}
In particular, it is strictly positive for all $ \mu > 0 $. This changes drastically in the regime of localization.
\begin{theorem}\label{thm:nosf} 
If condition (SULE) holds  (with $g_L(\alpha)=L$) for the energy regime $(-\infty,\mu ]$, then for any $ \theta > 0 $ and almost surely:
\begin{equation}
\limsup_{L\to \infty} L \left( E_L(N_\mu,\theta) - E_L(N_\mu, 0)  \right)  = 0 \, .  
\end{equation}
As a consequence, the superfluid density~\eqref{eq:sfd} is zero almost surely.
\end{theorem}

Our result implies that generically a disordered Bose gas in one dimension in the Tonks-Girardeau regime shows no superfluidity and no BEC even at zero temperature. This statement concerns the usual thermodynamic limit. We note that other limiting regimes are possible, corresponding to mean-field type interactions, where both BEC and superfluidity can prevail at zero temperature \cite{SYZ,KMSY} (see also  \cite{BW,SMBW,KV14} for related results). The proof of Thm.~\ref{thm:nosf} will be given in Section~\ref{sec:nosfl}.

\section{Proof of decay of correlations}\label{sec:DecCor}
This section is devoted to the proofs of Theorem~\ref{thm:1pd} and Corollary~\ref{cor:nobec}. Since  $ N $ is kept fixed in Thm.~\ref{thm:1pd}, we will drop the superscript $ \sharp_N $ on the eigenfunctions $ \{\varphi_{j,L}^{\sharp_N}\}$ of $ H_L^{\sharp_N} $ (as well as their dependence on $ \omega $) for ease of notation.\\

Using the Laplace formula, the determinantal expression~\eqref{eq:opddet} for the kernel of the one-particle reduced density matrix can be recast as
\begin{equation}\label{eq:1pdadj}
\gamma_\Psi(x,y) = \sum_{\alpha,\beta}  \varphi_{j_\alpha,L}(x)  \overline{\varphi_{j_\beta,L}(y)} \, \left[ \adj K_N(x,y) \right]_{\beta,\alpha} =: \langle \varphi(y) \, , \, \adj K_N(x,y) \,  \varphi(x) \rangle
\end{equation}
where the last inner product is in $ \C^{N} $ and the adjugate matrix of $ K_N(x,y) $ is the matrix of cofactors (up to signs), i.e., 
\begin{equation}\label{eq:adjugate}
 \left[ \adj K_N(x,y) \right]_{\beta,\alpha}  = \det \begin{pmatrix} 0 & e_\alpha^T \\ e_\beta &  K_N(x,y) \end{pmatrix}
 = (-1)^{\alpha+\beta} \det[ K_N(x,y) ]_{\hat\alpha,\hat\beta} \, .
\end{equation}
Here $ \{e_\alpha \}$ denote the unit vectors and the hats indicate the deletion of row $ \beta $ and column $ \alpha $ from $ K_N(x,y) $. 
Some key properties are summarized in the following:
\begin{enumerate}
\item In case $ x = y $, $ K_N(x,x) = \adj K_N(x,x) $ equals the identity matrix.
\item  Since $ K_N(x,y) = P_N - 2 P_N 1_{[x,y]} P_N $, the  inequality $-P_N \leq  K_N(x,y) \leq P_N $ holds, and hence we arrive at the bound
$  \| K_N(x,y)  \| \leq 1 $
 on the operator norm.
\item Since the adjugate matrix of any hermitian $ N \times N $ matrix is hermitian with eigenvalues given by the products of $ N-1 $ disjoint eigenvalues of the matrix, the norm bound 
\begin{equation}\label{eq:adj}
\|  \adj K_N(x,y) \| \leq 1 
\end{equation} 
is an immediate consequence of $ \| K_N(x,y)  \| \leq 1 $. 
\end{enumerate}

Our basic strategy for an estimate of~\eqref{eq:1pdadj} is to split the summation depending on whether the eigenstates  live predominantly to the right or left of the midpoint of $ n$ and $m$. To to so, 
we will suppose without loss of generality $1\leq n <  m\leq L/2$ and abbreviate by 
\begin{equation}\label{def:M}
 M := \lfloor (m+n)/2   \rfloor
 \end{equation}
 the midpoint between $ n $ and $ m $ (or between $n$ and $m-1$ if $(n+m)/2$ is not an integer). This midpoint introduces a left/right partition of the system according to which we may sort 
the eigenstates:
\begin{align}
\mathcal{L}  := \left\{ \alpha \, \Big| \, \int_0^M |\varphi_{j_\alpha,L}(\xi) |^2 > \frac{1}{2} \right\}  \, , \quad
\mathcal{R}   := \left\{ \alpha \, \Big| \, \int_M^L |\varphi_{j_\alpha,L}(\xi) |^2  \geq  \frac{1}{2} \right\} \, . 
\end{align}
The normalization of eigenstates implies that $ \mathcal{L} $ and $ \mathcal{R} $ constitute a disjoint partition of the finite index set $ \{1,\dots,N\}$. 
Writing the vectors on the right side of \eqref{eq:1pdadj} accordingly as  $  \varphi(x) =  \varphi_{\mathcal{L}}(x) +  \varphi_{\mathcal{R}}(x)   $ (and similarly for $ \varphi(y) $) we may 
split the sum~\eqref{eq:1pdadj}  into three parts, $ \gamma_\Psi(x,y) = \gamma_\Psi^{(1)}(x,y) + \gamma_\Psi^{(2)}(x,y) + \gamma_\Psi^{(3)}(x,y) $, with 
\begin{align}\label{eq:decomposition} 
 \gamma_\Psi^{(1)}(x,y) & := \langle  \varphi_{\mathcal{L}}(y) \, , \,  \adj K_N(x,y) \,  \varphi(x)  \rangle \, , \notag \\
   \gamma_\Psi^{(2)}(x,y) & := \langle  \varphi_{\mathcal{R}}(y) \, , \,  \adj K_N(x,y) \,  \varphi_{\mathcal{R}}(x)  \rangle   \notag \\
 \gamma_\Psi^{(3)}(x,y)  & := \langle  \varphi_{\mathcal{R}}(y) \, , \,  \adj K_N(x,y) \,  \varphi_{\mathcal{L}}(x)  \rangle \, .
\end{align}
The Hilbert-Schmidt norms of these contributions are estimated separately. We start with the first two terms.
\begin{lemma}\label{lem:1}
For all $ m > n $:
\begin{align}
\| 1_{I_n}  \gamma_\Psi^{(1)} 1_{I_m}  \|_2  & \leq  \sqrt{2\, N(I_n)}  \ \sum_{k =1 }^{M} Q_L(k,m;J)  \notag \\
\| 1_{I_n}  \gamma_\Psi^{(2)} 1_{I_m}  \|_2  & \leq   \sqrt{2\, N(I_m)}   \ \sum_{k =M+1}^L Q_L(n,k;J) \, , 
\end{align}
where $ M $ was defined in~\eqref{def:M} and $ N(I_k) := \int_{I_k} P_N(x,x) \, dx $ denotes the local particle number in $ I_k $. 
\end{lemma}
\begin{proof}
The Cauchy-Schwarz inequality in $ \mathbb{C}^N $ and~\eqref{eq:adj} imply
\begin{equation}
 \left|  \gamma_\Psi^{(1)}(x,y)\right|^2 \  \leq \|  \varphi_{\mathcal{L}}(y) \|^2 \, \| \varphi(x) \|^2 \, \|  \adj K_N(x,y) \|^2  \leq P_N(x,x) \, \sum_{\alpha \in \mathcal{L}} |\varphi_{j_\alpha,L}(y) |^2 \, .
 \end{equation}
Integration over $ x \in I_n $ and $ y \in I_m $ yields the bound 
\begin{equation}
 \| 1_{I_n}  \gamma_\Psi^{(1)} 1_{I_m}  \|_2^2 \leq N(I_n) \sum_{\alpha \in \mathcal{L}} \Phi_{j_\alpha,L}(m)^2 \, .
 \end{equation}
The last term may be estimated using the fact that eigenfunctions corresponding to $ \alpha \in \mathcal{L} $ predominantly live on the left:
\begin{equation}\label{eq:linkssupp}
 \sum_{\alpha \in \mathcal{L}} \Phi_{j_\alpha,L}(m)^2 \leq 2 \sum_{k =1 }^{M} \sum_{\alpha\in \mathcal{L}} \Phi_{j_\alpha,L}(k)^2 \Phi_{j_\alpha,L}(m)^2 \leq 2 \sum_{k =1 }^{M} Q_L(k,m;J)^2 \, . 
\end{equation}
Finally, we use that 
\begin{equation}
\sum_{k =1 }^{M} Q_L(k,m;J)^2 \leq \left(\sum_{k =1 }^{M} Q_L(k,m;J) \right)^2 \,.
\end{equation}
This completes the proof of the first inequality. A proof of the second inequality proceeds analogously with the roles of $ \mathcal{L} $ and $ \mathcal{R} $ interchanged.  
\end{proof}

Both terms  are thus (sub-)exponentially small in the distance between $ n $ and $ m $ provided the eigenfunction correlator decays accordingly. To establish this for the third term we 
employ the technique developed in \cite{SW} for estimates on certain structured determinants.
\begin{lemma}\label{lem:2}
For all $ m > n $:
\begin{equation}\label{eq:lem2}
\| 1_{I_n}  \gamma_\Psi^{(3)} 1_{I_m}  \|_2 \leq  2 
\sqrt{ 2\, e N(I_n) \, N(I_m)} \, \sum\nolimits_{(k,l) }'  \sqrt{N(I_k)} \,  Q_L(k,l;J)\, . 
\end{equation}
Here we have abbreviated $ \sum_{(k,l) }' :=  \left[  \sum_{k=1}^n   \sum_{l=M+1}^L  + \sum_{k=m}^L   \sum_{l=1}^M \right] $. 
\end{lemma}
\begin{proof}
Using the definition of the adjugate matrix, we rewrite the inner product again as a determinant of a block matrix of the following form:
\begin{equation}\label{eq:g3}
 \gamma_\Psi^{(3)}(x,y)  =  \det  \begin{pmatrix} 0 & \varphi_{\mathcal{L}}(x)^T  \quad 0  \\ \begin{matrix} 0  \\ \overline{\varphi_{\mathcal{R}}(y)} \end{matrix} &  \widetilde K_N(x,y) \end{pmatrix}
\end{equation}
where $\widetilde K_N(x,y) = V^T K_N(x,y) V$ for a permutation matrix $V$ that permutes the indices such that the indices in $\mathcal{L}$ correspond to the first $|\mathcal{L}|$ rows and columns, and the ones in the $\mathcal{R}$ to the last $|\mathcal{R}| = N - |\mathcal{L}|$ ones. Note that $\| \widetilde K_N(x,y) \| \leq 1 $. Hence we can apply the following estimate on the determinant, which is a simple variant of the bound in Theorem~3.1 in \cite{SW}.

\begin{lemma}\label{lem:3}
Let $v\in\C^p$, $w\in \C^q$, and let $K = \begin{pmatrix}A & B \\ C & D \end{pmatrix}$ be a $(p+q)\times (p+q)$ matrix with $\| K\|\leq 1$. Then
\begin{equation}\label{eq:lem3}
\left | \det  \begin{pmatrix} 0 & v^T &  0 \\ 0 & A & B  \\ w  &  C & D  \end{pmatrix} \right| \leq \sqrt{e} \|v\| \|w\| \|B\| \,.
\end{equation}
\end{lemma}

\begin{proof} 
By linearity we may assume that $\|w\|=1$. As in \cite{SW}, we first apply a unitary  operator $U$ on $\C^p$ that takes $v$ into the vector $(0,\dots,0,\|v\|)$. Moreover, we can find another unitary $V$ on $\C^p$ such that $V A U^T$ is upper triangular, i.e., all entries below the diagonal are zero. The left side of \eqref{eq:lem3} is equal to the absolute value of the determinant of 
\begin{equation}
\begin{pmatrix} 1 & 0 &  0 \\ 0 & V & 0  \\ 0  &  0 & 1_q  \end{pmatrix} \begin{pmatrix} 0 & v^T &  0 \\ 0 & A & B  \\ w  &  C & D  \end{pmatrix} \begin{pmatrix} 1 & 0 &  0 \\ 0 & U^T & 0  \\ 0  &  0 & 1_q  \end{pmatrix} = \begin{pmatrix} 0 & (Uv)^T &  0 \\ 0 & V AU^T & V B  \\ w  &  C U^T & D  \end{pmatrix} =:M \,.
\end{equation}
To estimate it, we use Hadamard's bound, which states that the determinant of a matrix is bounded by the product of the norms of the row vectors. Before we apply this bound, we perform one more operation that leaves the determinant invariant, namely we subtract $s$ times the $(p+1)^{\rm th}$ row from the first row, for some $s\in \C$. Let $\alpha$ denote the lower right entry of $V A U^T$ (i.e., the only non-zero entry in the $p^{\rm th}$ row of $VAU^T$). Using the fact that the norm of a row vector of a square matrix can never exceed the norm of the matrix,  we then obtain
\begin{equation}
\left| \det M \right| \leq \sqrt{  \left|   \|v\| - s \alpha  \right|^2 + |s|^2 \|B\|^2 }  \sqrt{ |\alpha|^2 + \|B\|^2 }  \prod_{\alpha=1}^q \sqrt{ 1 + |w_\alpha|^2 } \,.
\end{equation}
The first factor on the right side bounds the norm of the first row, the second one the $(p+1)^{\rm th}$ row, and the last factors the rows $p+2,\dots p+q+1$. The norms of the other rows $2,\dots,p$ are bounded by one. Since 
\begin{equation}
\prod_{\alpha=1}^q \sqrt{ 1 + |w_\alpha|^2 } = \exp\left( \frac 12 \sum_{\alpha=1}^q \ln\left( 1 + |w_\alpha|^2 \right) \right) \leq  \exp\left( \frac 12 \sum_{\alpha=1}^q  |w_\alpha|^2 \right) =  \sqrt{e}  \,,
\end{equation}
the choice $ s = \overline{\alpha} \|v\| ( |\alpha|^2 + \|B\|^2)^{-1} $ leads to the desired bound \eqref{eq:lem3}.
\end{proof}

An application of Lemma~\ref{lem:3} to \eqref{eq:g3} leads to the bound 
\begin{equation}\label{eq:bound3b}
\left|  \gamma_\Psi^{(3)}(x,y)  \right| \leq \sqrt{e} \|  \varphi_{\mathcal{R}}(y) \| \, \| \varphi_{\mathcal{L}}(x) \| \, \| B \| \, ,
\end{equation}
where $ B $ is the $|\mathcal{L}| \times |\mathcal{R}|$ matrix  with entries  $ \alpha \in \mathcal{L} $, $\beta \in \mathcal{R} $ given by
\begin{equation}
 B_{\alpha,\beta} := [ K_N(x,y)]_{\alpha,\beta}  =  2 \sum_{k=1}^L \, \int_{I_k\cap [x,y]^c}  \varphi_{j_\alpha,L}(z) \overline{\varphi_{j_\beta,L}(z) } dz\, . 
 \end{equation}
Here the equality results from~\eqref{def:KN}. The operator norm of $ B $ is estimated in terms of its Frobenius norm, $ \| B \| \leq \| B \|_2 $, which in turn is bounded from above using Minkowski's inequality as follows. For $x\in I_n$, $y\in I_m$, 
\begin{align}
 \| B \|_2 \ & = \Bigg( \sum_{\substack{ \alpha \in \mathcal{L}  \\ \beta \in \mathcal{R} }} |  B_{\alpha,\beta} |^2 \Bigg)^{\frac{1}{2}}  \leq 2 \sum_{k=1}^L   \Bigg(  \sum_{\substack{ \alpha \in \mathcal{L}  \\ \beta \in \mathcal{R} }}  \left|  \int_{I_k\cap [x,y]^c}  \varphi_{j_\alpha,L}(z) \overline{\varphi_{j_\beta,L}(z) } dz \right|^2 \Bigg)^{\frac{1}{2}} \notag \\
& \leq 2 \left( \sum_{k=1}^n  + \sum_{k=m}^L \right) \left( \sum_{\alpha\in \mathcal{L}}  \Phi_{j_\alpha,L}(k)^2 \sum_{\beta\in \mathcal{R}}  \Phi_{j_\beta,L}(k)^2 \right)^\frac{1}{2} \notag \\
& \leq 2 \sum_{k=1}^n \sqrt{ N(I_k)} \left(2 \sum_{l=M+1}^L Q_L(k,l;J)^2 \right)^{\frac{1}{2}} + 2 \sum_{k=m}^L  \sqrt{ N(I_k)} \left( 2\sum_{l=1}^M Q_L(k,l;J)^2 \right)^{\frac{1}{2}} \notag \\ & \leq 2 \sqrt 2 \sum\nolimits_{(k,l)} '  \sqrt{N(I_k)} \,  Q_L(k,l;J)\, . 
\end{align}
The penultimate inequality derives from~\eqref{eq:linkssupp} (and its analog  for $ \mathcal{R} $). 
After inserting this bound in \eqref{eq:bound3b} and  integrating its square over $ x \in I_n $ and $ y \in I_m $, we obtain the claimed bound \eqref{eq:lem2}.
\end{proof}

We may now conclude the proof of our first main result.

\begin{proof}[Proof of Theorem~\ref{thm:1pd}]
We start by noting that $\mathbb{E}[\| 1_{I_n} \gamma_\Psi 1_{I_m} \|_2]$ is uniformly bounded. In fact, the Cauchy-Schwarz inequality and 
the fact that the Hilbert-Schmidt norm is dominated by the trace norm lead to 
\begin{equation}\label{eq:Nb}
\| 1_{I_n} \gamma_\Psi 1_{I_m} \|_2 \leq  \left( \| 1_{I_n} \gamma_\Psi 1_{I_n} \|_2 \| 1_{I_m} \gamma_\Psi 1_{I_m} \|_2 \right)^{1/2} \leq \sqrt{   N(I_n) N(I_m)}  \, . 
\end{equation}
Moreover, by the Cauchy-Schwarz inequality for the expectation value,
\begin{equation}
 \mathbb{E}[\sqrt{   N(I_n) N(I_m)}] \leq \left(  \mathbb{E}[N(I_n)]  \mathbb{E}[N(I_m)]  \right)^{1/2}\,.
 \end{equation}
In turn, the average local particle number is uniformly bounded by assumption:
\begin{equation}\label{eq:particlenumberbound}
\sup_{L, n} \;\mathbb{E}\left[N(I_n) \right] = \sup_{L, n} \;  \mathbb{E}\left[Q_L(n,n;J ) \right] \leq  C \, .
\end{equation}
We may therefore assume without loss of generality that $ m > n $. We may also assume that $1\leq n < m\leq L/2$; the general case then follows by a simple relabeling. We first proof the assertion for $ \sigma = \frac{2}{5} $. Since 
\begin{equation}
\mathbb{E}\Big[ \| 1_{I_n}  \gamma_\Psi 1_{I_m}  \|_2^{\frac{2}{5}} \Big]  \leq  \sum_{i=1}^3 \mathbb{E}\Big[ \| 1_{I_n}  \gamma_\Psi^{(i)} 1_{I_m}  \|_2^{\frac{2}{5}} \Big]
\end{equation}
we can treat the three contributions to $\gamma_\Psi$ separately. For the first term, 
 Lemma~\ref{lem:1} and the H\"older inequality for the expectation value imply
\begin{equation}
\mathbb{E}\Big[ \| 1_{I_n}  \gamma_\Psi^{(1)} 1_{I_m}  \|_2^{\frac{2}{5}} \Big]  \leq  2^\frac{1}{5}  \; \mathbb{E}\left[N(I_n)\right]^{\frac{1}{5} }   
\Big(\sum\nolimits_{k =1 }^{M} \mathbb{E}\left[Q_L(k,m;J)\right] \Big)^{\frac{2}{5}}     \, . 
\end{equation}
The first factor is uniformly bounded according to~\eqref{eq:particlenumberbound}. The last factor is bounded using the localization assumption (ECL):
\begin{equation}
 \sum_{k =1 }^{M} \mathbb{E}\left[Q_L(k,m;J)\right] \leq  C \, \sum_{k =1 }^{M} \exp\left(-\frac{\dist(k,m)^\xi}{\ell^\xi} \right) \leq C'  \exp\left( - \frac{\dist(M,m)^\xi}{\ell^\xi} \right) \, . 
\end{equation}
Since $ \dist(M,m) \geq \tfrac 12 \dist(m,n)$ this implies the claim for the first term in the decomposition~\eqref{eq:decomposition}. The second term is treated similarly. 
For the third term, we employ Lemma~\ref{lem:2} and H\"older's inequality for the expectation value to conclude:
\begin{align}
& \notag \mathbb{E}\Big[ \| 1_{I_n}  \gamma_\Psi^{(3)} 1_{I_m}  \|_2^{\frac{2}{5}} \Big]  \\ & \leq  e^{1/5} 2^{3/5}    \,  \mathbb{E}\left[N(I_n)\right]^{\frac{1}{5} } \mathbb{E}\left[ N(I_m)\right]^{\frac{1}{5}}  \mathbb{E} \left[ \left( \sum\nolimits_{(k,l) }'  \sqrt{N(I_k)} Q_L(k,l;J) \right)^{2/3} \right]^{3/5} \notag \\
& \leq e^{1/5} 2^{3/5}    \,   \mathbb{E}\left[N(I_n)\right]^{\frac{1}{5} } \mathbb{E}\left[ N(I_m)\right]^{\frac{1}{5}}  \sum\nolimits_{(k,l) }' \,   \mathbb{E}\left[N(I_k)\right]^{\frac{1}{5} }  \mathbb{E}\left[Q_L(k,l;J)\right]^{\frac{2}{5} } \, .     \end{align}
The terms involving the local particle number are uniformly bounded. The last term is again bounded using the localization assumption (ECL). In fact, \begin{equation}
 \sum\nolimits_{(k,l) }'   \exp\left(-\frac{2\dist(k,l)^\xi}{5\ell^\xi} \right) \leq C \exp\left( - \frac{2\min\{\dist(M,m)^\xi, \dist(M,n)\}^\xi}{5\ell^\xi} \right) \, .
\end{equation}
Since the distance of the midpoint $ M$ to either $ n $ and $ m$ is at most $\frac 12 (1+  \dist(n,m))$, this completes the proof in case $ \sigma = \frac{2}{5} $.

The general case follows with the help of interpolation from the bounds~\eqref{eq:Nb}--\eqref{eq:particlenumberbound}: For $2/5< \sigma < p$, 
\begin{align}\label{eq:interpol}
\mathbb{E}\left[ \| 1_{I_n} \gamma_\Psi 1_{I_m} \|_2^\sigma \right] &  \leq \mathbb{E}\left[ \| 1_{I_n} \gamma_\Psi 1_{I_m} \|_2^{\frac{2 (p-\sigma)}{5p-2}} N(I_n)^{\frac{p}{2}\frac{5\sigma-2}{5p-2}}N(I_m)^{\frac{p}{2}\frac{5\sigma-2}{5p-2}}  \right]  \notag \\
& \leq \mathbb{E}\left[ \| 1_{I_n} \gamma_\Psi 1_{I_m} \|_2^\frac{2}{5}\right]^{\frac{5( p-\sigma)}{5p-2}} \mathbb{E}\left[N(I_n)^p \right]^{\frac{5\sigma-2}{2(5p-2)}}  \mathbb{E}\left[N(I_m)^p \right]^{\frac{5\sigma-2}{2(5p-2)}} \,,
\end{align}
where the second step is  H\"older's inequality for the expectation. 
Since the last two factors are uniformly bounded for $p=1$ we arrive at the claim.
\end{proof}

A proof of absence of BEC then  proceeds as follows:
\begin{proof}[Proof of Corollary~\ref{cor:nobec}]
Assumption~\eqref{eq:asscor} implies that $\sup_{L,n}  \mathbb{E}[N(I_n)^p ] < \infty $.
From the interpolation bound in~\eqref{eq:interpol} and \eqref{eq:1pd} we conclude that for any  $2/5< \sigma < p$
\begin{equation}\label{eq:boundp}
\mathbb{E}\left[ \| 1_{I_n} \gamma_\Psi 1_{I_m} \|_2^\sigma \right] \leq C \,   \exp\left(-\frac{2( p-\sigma)}{5p-2} \frac{\dist(n,m)^\xi}{(2\ell)^\xi} \right)  
\end{equation}
with some $ C < \infty $ that is independent of $ L $, $n$ and $m$. Since $\|\gamma_\Psi\| \leq \max_n \sum_m \| 1_{I_n} \gamma_\Psi 1_{I_m} \|$, we can bound
\begin{equation}\label{eq:mink}
\mathbb{E}\left[ \|  \gamma_\Psi \|^\sigma \right]  \leq \sum_n   \mathbb{E}\left[ \left( \sum_{m} \| 1_{I_n} \gamma_\Psi 1_{I_m} \| \right)^{\sigma} 
\right]  \leq    \sum_{n} \left( \sum_m \mathbb{E}\left[  \| 1_{I_n} \gamma_\Psi 1_{I_m} \|^\sigma  \right]^{1/\sigma} \right)^{\sigma}  ,
\end{equation}
where in the last step we used Minkowski's inequality for the expectation. Since the operator norm is bounded by the Hilbert-Schmidt norm, \eqref{eq:boundp} implies that the sum over $m$ in \eqref{eq:mink} is bounded, independently of $n$. This shows $\mathbb{E}\left[ \|  \gamma_\Psi \|^\sigma \right]  \leq C L$
with some $ C < \infty $ that  is independent of $ L $ and $N$. A Chebychev estimate then implies for any $ \varepsilon > 0 $ and $r>0$ 
\begin{equation}\label{eq:cheb2}
\mathbb{P}\left( \|  \gamma_\Psi \| > \varepsilon L^r \right) \leq \frac{\mathbb{E}\left[ \|  \gamma_\Psi \|^\sigma \right]} {\varepsilon^\sigma L^{r\sigma} } \leq \frac{C}{\varepsilon^\sigma}  \, L^{1-\sigma r} \, ,
\end{equation}
where $\mathbb{P}$ stands for the probability of an event. 
If we choose $r >  2/\sigma$, then  $ 1 - \sigma r < - 1 $, and the right side of \eqref{eq:cheb2} is summable in  $ L $. The 
Borel-Cantelli lemma thus yields the claimed almost-sure convergence. 
\end{proof}

\section{Proof of the absence of superfluidity}\label{sec:nosfl}
For a proof of Theorem~\ref{thm:nosf}, let $H_L(\theta)$ denote the self-adjoint operator which acts as~\eqref{def:Ham} on functions with {\em twisted} boundary conditions,  $ \psi(L) = e^{i\theta} \psi(0) $ and $ \psi'(L) = e^{i\theta} \psi'(0) $. Let $\{ E_{j,L}(\theta)\}$ denote its eigenvalues, ordered increasingly with $j$, i.e., $E_{j,L}(\theta) \leq E_{j+1,L}(\theta)$. 
With 
\begin{equation}
\theta_\mu := \theta + \frac \pi 2 \left ( 1 + (-1)^{N_\mu} \right) 
\end{equation}
we have 
\begin{equation}
E_L(N_\mu,\theta) = \sum_{j =1}^{N_\mu} E_{j,L}(\theta_\mu) \,. 
\end{equation}
We claim that $E_{N_\mu,L}(\theta_\mu) \leq \mu$. This follows from the fact that, by construction $E^\pm_{N_\mu,L}\leq \mu$, and $E_{N_\mu,L}(\theta_\mu) \leq \max \{ E^+_{N_\mu,L},E^-_{N_\mu,L}\}$. Hence we can invoke the variational principle in the form
\begin{equation}\label{eq:varpr} 
E_L(N_\mu,\theta) = \mu N_\mu + \inf \left\{  \tr  [H_L(\theta_\mu)-\mu]\gamma \, \big| \, 0 \leq \gamma\leq 1\, , \ \tr \gamma\leq N_\mu\right\}\,.
\end{equation}
We emphasize that it is possible here to relax the condition $\tr \gamma = N_\mu$ to $\tr \gamma \leq N_\mu$ exactly because $E_{N_\mu,L}(\theta_\mu) \leq \mu$. This turns out to  be convenient  in the following.

To obtain an upper bound on $E_L(N_\mu,\theta)$, and hence on $ \rho_s $, we choose as a trial density matrix in \eqref{eq:varpr} 
\begin{equation}
 \gamma = \frac{ \tilde \gamma}{\max\{ \| \tilde \gamma \|, 1\} }  \, , \quad \mbox{with}\quad \tilde \gamma := \sum_{j : \, E_{j,L} ^{\sharp_\mu} \leq \mu} e^{i\psi_{j,L}} |\varphi_{j,L} ^{\sharp_\mu}\rangle \langle  \varphi_{j,L} ^{\sharp_\mu} | e^{-i\psi_{j,L}} \,.
\end{equation}
Here,  $ \{\varphi_{j,L} ^{\sharp_\mu} \} $ abbreviates an orthonormal eigenbasis of $ H_L ^{\sharp_\mu}$ and $ \{ E_{j,L} ^{\sharp_\mu} \} $ are the corresponding eigenvalues. Note that $\tr \tilde\gamma = N_\mu$, as remarked in \eqref{eq:remnu}, hence $0\leq \gamma\leq 1$ and $\tr \gamma \leq N_\mu$.

The  trial phase functions $ \psi_{j,L} : [0,L] \to \mathbb{R} $ will be chosen continuous, increasing (and piecewise differentiable) such that $ \psi_{j,L}(0) = 0 $ and $ \psi_{j,L}(L) = \theta $. We pick them depending on an additional variational parameter $ \delta > 0 $ to be chosen later. 
More specifically, we set 
\begin{equation}\label{def:I}
 I_j(\delta) := \bigcup_{k: \,  \Phi_{j,L} ^{\sharp_\mu}(k)  \leq \delta } I_k \, , 
  \end{equation}
 and choose 
 \begin{equation}
\tilde  I_j(\delta) \subseteq  I_j(\delta) 
\end{equation}
 to be the largest connected subset of $ I_j(\delta)$.  We then simply take $\psi_{j,L}$ to increase linearly on $\tilde  I_j(\delta)$ with slope $\theta/|\tilde  I_j(\delta)|$, and constant otherwise. 
Note that  $  \tilde I_j(\delta) $ is certainly non-empty   for $ \delta > L^{-1/2}  $, since $ 1 \geq \sum_{n :  \Phi_{j,L} ^{\sharp_\mu}(n) > \delta } \Phi_{j,L} ^{\sharp_\mu}(n)^2 \geq (L - |I_j(\delta)| ) \, \delta^2 $ which implies $  | I_j(\delta) | \geq L - \delta^{-2} $. The localization assumption (SULE) may be used for a stronger estimate. Namely,~\eqref{eq:locass} implies that  $ \Phi_{j,L}^{\sharp_\mu}(k) \leq \delta $ for any $ k \in \{1, \dots , L\} $ with $\dist (k, \gamma_{j,L}^{\sharp_\mu})^\xi  \geq \ell^\xi \ln( A_L L^{3/2}/\delta)$.  Choosing 
\begin{equation}\label{eq:defdelta}
\delta := A_L \,  L^{3/2} \exp\left( - \frac{L^\xi}{(4\ell)^\xi} \right) \, , 
\end{equation}
we see that $ \Phi_{j,L} ^{\sharp_\mu}(k) \leq \delta $ whenever   $\dist (k, \gamma_{j,L}^{\sharp_\mu}) \geq L/4$, and thus $I_j(\delta)$ contains an interval of length at least $(L-4)/4$. Therefore, 
\begin{equation}\label{eq:locestI}
 |\tilde I_j(\delta)| \geq  \frac{L-4}{4} \, . 
\end{equation}

A straightforward computation shows that
\begin{align}
& \tr [H_L(\theta_\mu)-\mu] \gamma  \notag  \\ & =  - \frac{ \tr [H_L^{\sharp_\mu} -\mu]_-}{\max\{ \| \tilde \gamma \|, 1\} } +   \frac{ 1}{\max\{ \| \tilde \gamma \|, 1\} } \sum_{j : \, E_{j,L} ^{\sharp_\mu} \leq \mu}  \int_0^L \left| \varphi_{j,L} ^{\sharp_\mu}(x)\right|^2  \left| \frac{d}{dx} \psi_{j,L}(x)\right|^2 dx \notag \\
 & \leq  - \frac{ \tr [H_L^{\sharp_\mu} -\mu]_-}{\max\{ \| \tilde \gamma \|, 1\} } + \theta^2  \sum_{j : \, E_{j,L} ^{\sharp_\mu} \leq \mu}  \frac{1}{|\tilde I_j(\delta)|^2} \sum_{ n: \,  I_n \subset \tilde I_j(\delta) } \Phi_{j,L} ^{\sharp_\mu}(n)^2 \notag \\
 & \leq - \frac{ \tr [H_L^{\sharp_\mu} -\mu]_-}{\max\{ \| \tilde \gamma \|, 1\} } + \theta^2 \delta^2   \sum_{j : \, E_{j,L} ^{\sharp_\mu} \leq \mu}  \frac{1}{|\tilde I_j(\delta)|} \, .  \label{eq:fup}
\end{align}
To estimate the norm of $\tilde \gamma$, we note that $\tilde\gamma$ is unitarily equivalent to the $N_\mu\times N_\mu$ matrix with matrix elements $\langle \varphi_{j,L} ^{\sharp_\mu} e^{i\psi_{j,L}} ,e^{i\psi_{k,L}} \varphi_{k,L} ^{\sharp_\mu}\rangle$. 
In particular, 
\begin{equation}\label{eq:tg}
\|\tilde\gamma\| \leq \max_{j} \sum_k \left| \langle \varphi_{j,L} ^{\sharp_\mu} e^{i\psi_{j,L}} , e^{i\psi_{k,L}} \varphi_{k,L} ^{\sharp_\mu}\rangle \right| \,.
\end{equation}
For $j\neq k$, we have
\begin{align}
&  \int_{\tilde I_j(\delta)^c \cap \tilde I_k(\delta)^c} \overline{\varphi_{j,L} ^{\sharp_\mu}(z)} e^{-i ( \psi_{j,L}(z) - \psi_{k,L}(z))} \varphi_{k,L} ^{\sharp_\mu}(z) dz \notag \\ &  \notag = e^{i\alpha} \int_{\tilde I_j(\delta)^c \cap \tilde I_k(\delta)^c} \overline{\varphi_{j,L} ^{\sharp_\mu}(z)}  \varphi_{k,L} ^{\sharp_\mu}(z) dz \\ &= - e^{i\alpha}  \int_{\tilde I_j(\delta) \cup \tilde I_k(\delta)} \overline{\varphi_{j,L} ^{\sharp_\mu}(z)}  \varphi_{k,L} ^{\sharp_\mu}(z) dz
\end{align}
since $\psi_{k,L}(z) - \psi_{j,L}(z)) = \alpha$ is a constant on $\tilde I_j(\delta)^c \cap \tilde I_k(\delta)^c$. Hence 
\begin{align}\notag
& \left| \langle \varphi_{j,L} ^{\sharp_\mu} e^{i\psi_{j,L}} , e^{i\psi_{k,L}} \varphi_{k,L} ^{\sharp_\mu}\rangle \right| \\ & \notag \leq  \int_{\tilde I_j(\delta) \cup \tilde I_k(\delta)} \left| e^{-i ( \psi_{j,L}(z) - \psi_{k,L}(z))} - e^{i\alpha}\right|  | \varphi_{j,L} ^{\sharp_\mu}(z)  \varphi_{k,L} ^{\sharp_\mu}(z) | dz \\ & \leq 2 \max_{|\beta| \leq |\theta|} |1-e^{i\beta}|  \delta L \leq 2 |\theta| \delta L  
\end{align}
for $j\neq k$ and $|\theta|\leq \pi$. In combination with \eqref{eq:tg}, this implies 
\begin{equation}\label{eq:tg2}
\|\tilde\gamma\| \leq 1 + 2 |\theta| \delta L N_\mu \,.
\end{equation}

Inserting the bounds~\eqref{eq:locestI} and \eqref{eq:tg2} in \eqref{eq:fup}, and using \eqref{eq:varpr} and \eqref{eq:remnu},  we hence conclude that for any $ L \geq 5 $
\begin{equation}\label{eq:lbxl}
 X_L := L \ \frac{ E_L(N_\mu, \theta) - E_L(N_\mu,0)  }{\theta^2} \leq \frac {2\delta L^2}{|\theta|}  \tr [H_L^{\sharp_\mu} -\mu]_- N_\mu  + \frac {4 L \delta^2}{L-4}N_\mu  \, ,
\end{equation}
with $\delta$ given in \eqref{eq:defdelta}. Our goal is to show that $\lim_{L\to\infty} X_L =0$ almost surely. Note that $\delta$ is random, but the coefficient $A_L$ in \eqref{eq:defdelta} is uniformly bounded in expectation, according to the assumption \eqref{eq:unifA}. Moreover, we have the following rough but uniform bounds:

\begin{lemma}\label{lem:4}
Under assumption \eqref{eq:assv} one has, irrespective of boundary conditions, 
\begin{equation}\label{eq:TEA}
\sup_{L\in \mathbb{N} } \frac{\mathbb{E}\left[N_\mu \right] }{L} < \infty \quad \text{and} \quad \sup_{L\in \mathbb{N} } \frac{\mathbb{E}\left[ \big(  \tr [H_L^{\sharp_\mu} -\mu]_-\big)^{1/3} \right]}{L} < \infty \,.
\end{equation}
\end{lemma}
\begin{proof}
With $h_n$ denoting the restriction of $H_L$ to $I_n$, with Neumann boundary conditions, it is well known \cite[Sect.~XIII.15]{RS4} that 
\begin{equation}\label{eq:nln}
N_\mu \leq \sum_{n=1}^L \tr 1_{(-\infty,\mu]}(h_n)\,.
\end{equation}
A simple calculation based on the Birman-Schwinger principle (see, e.g., \cite[Chap.~7]{SimonTrace}) shows that  $\infspec h_n$ can be bounded in terms of $\int_{I_n} |V|$, which has a finite expectation according to our assumption (\ref{eq:assv}). In fact, one has 
\begin{equation}\label{eq:infspechn}
\infspec h_n \geq -  f^{-1} ( \mbox{$ \int_{I_n} |V_-| $})  \quad , \ \quad f(t) = \sqrt{t} \tanh \sqrt{t} \,,
\end{equation}
where $V_-$ denotes the negative part of $V$. Keeping half of the kinetic energy, one also obtains
\begin{equation}
h_n \geq -\frac 12 \Delta_{I_n} - \frac 12 f^{-1}( 2 \mbox{$ \int_{I_n} |V_-| $})\,,
\end{equation}
and hence 
\begin{equation}\label{eq:trhn}
\tr 1_{(-\infty,\mu]}(h_n) \leq \tr 1_{[0,2\mu+  f^{-1}( 2 \mbox{$ \int_{I_n} |V_-| $})   ]}(-\Delta_{I_n})\,.
\end{equation}
It is easy to see that $\tr 1_{[0, \nu ]}(-\Delta_{I_n})$ grows like $\nu^{1/2}$ for large $\nu$, and $f^{-1}(t)$ grows like $t^2$ for large $t$. Hence \eqref{eq:trhn} is bounded by a constant times $1+ \int_{I_n} |V_-| $, which implies, in combination with \eqref{eq:nln}, the first bound in \eqref{eq:TEA}. 

To obtain the second, we use, similarly to \eqref{eq:nln}, that
\begin{equation}
\tr [H_L^{\sharp_\mu} -\mu]_- \leq  - \sum_{n=1}^L \tr (h_n-\mu) 1_{(-\infty,\mu]}(h_n)\,.
\end{equation}
In combination with \eqref{eq:infspechn} this implies 
\begin{equation}
\tr [H_L^{\sharp_\mu} -\mu]_- \leq \sum_{n=1}^L  \left( \mu + f^{-1} ( \mbox{$ \int_{I_n} |V_-| $}) \right) \tr 1_{(-\infty,\mu]}(h_n) \,.
\end{equation}
In particular, $(\tr [H_L^{\sharp_\mu} -\mu]_-)^{1/3}$ is bounded by a constant times $\sum_n ( 1 + \int_{I_n} |V_-|)$, which implies the desired result \eqref{eq:TEA}.
\end{proof}

Let us denote the right side of \eqref{eq:lbxl} by $Y_L$. From Lemma~\ref{lem:4}, \eqref{eq:unifA} and H\"older's inequality for the expectation value, it follows that 
\begin{equation}\label{eq:byl}
\mathbb{E}\left[Y_L^{1/5}\right]  \leq    C |\theta|^{-1/5} L^{3/2} \exp\left( - \frac 15 \frac{L^\xi}{(4\ell)^\xi} \right) 
\end{equation}
for $L\geq 5$ and some constant $C>0$ independent of $L$. 
The Chebychev inequality 
yields,   for any $ \varepsilon > 0 $, 
\begin{equation}
 \mathbb{P}\left( X_L \geq \varepsilon \right)  \leq  \mathbb{P}\left( Y_L \geq \varepsilon \right)  \leq \varepsilon^{-1/5}   \mathbb{E}\left[Y_L^{1/5} \right]  \, . 
\end{equation}
From \eqref{eq:byl},  the right side is seen to be summable in $ L$. The proof of Theorem~\ref{thm:nosf} is thus concluded with the help of the Borel-Cantelli lemma, which ensures that the probability that $ X_L \geq \varepsilon $ happens for infinitely many $ L \in \mathbb{N}  $ is zero. \hfill\qed

\subsection*{Acknowledgment}
Part of this work was carried out at the Erwin Schr\"odinger Institute (ESI) in Vienna, Austria, during the programme \lq\lq Quantum many-body systems, random matrices, and disorder\rq\rq. We thank the ESI for hospitality and financial support, and W. Zwerger for useful comments. R.S. acknowledges financial support by the Austrian Science Fund (FWF), project Nr. P 27533-N27, and S.W. by the Deutsche Forschungsgemeinschaft (WA 1699/2-1).

\appendix

\section{Appendix: Proof of dynamical properties of the density}\label{app}

 \begin{proof}[Proof of Proposition~\ref{prop:dyn}]
 By assumption on the range of the initial state, we have $  \Gamma_{ \phi_t} = U_t^*  \Gamma_{ \phi} U_t $  with $ U_t = e^{it H_L^{ \sharp_N}} P_J( H_L^{ \sharp_N}) $. Consequently,
 \begin{align}
 \left| \int_I  \varrho_t(x) dx -  \int_I \varrho_0(x) dx \right| & = \left| \tr 1_I U_t^*  \Gamma_{ \phi} U_t - \tr 1_I  \Gamma_{ \phi} \right|  \notag \\ 
 & = \left| \tr 1_I U_t^* 1_{I^c}  \Gamma_{\phi} U_t - \tr 1_{I^c} U_t^* 1_{I}  \Gamma_{\phi} U_t  \right| \notag \\
&  \leq \| 1_I U_t^* 1_{I^c} \|_1 + \| 1_{I^c} U_t^* 1_{I}  \|_1 \, , 
 \end{align}
 where the inequality follows from $ \|  \Gamma_{\phi} U_t  \| \leq \|  \Gamma_{\phi} \| \leq 1 $. 
 The first bound \eqref{eq:O1change} is then a consequence of 
 \begin{equation}\label{eq:dynloccons}
 \mathbb{E}\left[ \sup_{t\in \mathbb{R} } \| 1_{K^c} U_t^* 1_{I}  \|_ 1\right] \leq \sum_{\substack{ I_n \cap I \neq \emptyset \\ I_m  \cap K^c \neq \emptyset }}  \mathbb{E}\left[ \sup_{t\in \mathbb{R} } \| 1_{I_m} U_t^* 1_{I_n}  \|_ 1\right]  \, , 
 \end{equation} 
 valid for all $ I \subseteq K $. In case $ I = K $, the right side is bounded by a constant on account of~\eqref{eq:dynloc}; this concludes the proof of the first assertion.
 
 The proof of the second assertion \eqref{eq:nostatesrem} proceeds similarly. We estimate
  \begin{align}
  \int_I  \varrho_t(x) dx & =  \tr 1_I U_t^*  1_K \Gamma_{ \phi} 1_K U_t + \tr 1_I U_t^*  1_K \Gamma_{ \phi} 1_{K^c} U_t + \tr 1_I U_t^*  1_{K^c} \Gamma_{ \phi} U_t \notag \\
  & \leq \| 1_K \Gamma_{ \phi} 1_K \|_1 + \|  1_{K^c} U_t 1_I \|_1 + \| 1_I U_t^*  1_{K^c} \|_1 = \int_K \varrho_0(x) dx + 2 \|  1_{K^c} U_t 1_I \|_1 \, . 
  \end{align}
 The proof is completed using~\eqref{eq:dynloccons} and~\eqref{eq:dynloc}.
\end{proof}

\end{document}